\documentclass[11pt,a4paper,english]{article}

\newcommand{\beq}{\begin{equation}}
\newcommand{\eeq}{\end{equation}}

\usepackage[latin1]{inputenc}
\usepackage{babel}
\usepackage[centertags]{amsmath}
\usepackage{amsfonts}
\usepackage{amssymb}
\usepackage{color}
\usepackage{amsthm}
\usepackage{newlfont}
\usepackage{graphicx}
\usepackage{dsfont}

\DeclareMathOperator*{\essinf}{ess\,inf}

\newtheorem{theorem}{Theorem}[section]

\newtheorem{corollary}[theorem]{Corollary}
\newtheorem{proposition}[theorem]{Proposition}

\newtheorem{remark}[theorem]{Remark}
\newtheorem{Assumptions}[theorem]{Assumption}

\newcommand{\RR}{ \mathbb{R}}

\newcommand{\EE}{ \mathbb{E}}
\makeatletter  
\@addtoreset{equation}{section}

\makeatother  
\begin{document} 
\title{\textbf{Irreversible Investment under L\'{e}vy Uncertainty:\\an Equation for the Optimal Boundary}\footnote{Financial support by the German Research Foundation (DFG) via grant Ri 1128-4-1 is gratefully acknowledged by the first author. Research of the second author has been supported in part by a grant from Svenska kulturfonden via Stiftelsernas professorspool, Finland.}}
\author{Giorgio Ferrari\thanks{Center for Mathematical Economics, Bielefeld University, Germany; \texttt{giorgio.ferrari@uni-bielefeld.de}}
\and Paavo Salminen \thanks{{\AA}bo Akademi University, Department of Natural Sciences/Mathematics and Statistics, F\"{a}nriksgatan 3 B, FIN-20500 {\AA}bo, Finland; \texttt{phsalmin@abo.fi}}}
\date{\today}
\maketitle

\vspace{0.5cm}

{\textbf{Abstract.}} We derive a new equation for the optimal investment boundary of a general irreversible investment problem under exponential L\'{e}vy uncertainty.
The problem is set as an infinite time-horizon, two-dimensional degenerate singular stochastic control problem.
In line with the results recently obtained in a diffusive setting, we
show that the optimal boundary is intimately linked to the unique
optional solution of an appropriate Bank-El Karoui representation
problem. Such a relation and the Wiener-Hopf factorization allow us to
derive an integral equation for the optimal investment
boundary. In case the underlying  L\'{e}vy process hits any
  point in $ \RR$ with positive probability we show that the integral
  equation for the investment boundary is uniquely satisfied by the
  unique solution of another equation which is easier to handle. As a
  remarkable by-product we prove the continuity of the optimal investment boundary.
The paper is concluded with explicit results for profit functions of
(i) Cobb-Douglas type and (ii) CES type. In the first case the
function is separable  and in the second case non-separable.
\smallskip

{\textbf{Key words}}:
free-boundary, irreversible investment, singular stochastic control, optimal stopping, L\'{e}vy process, Bank and El Karoui's representation theorem, base capacity.

\smallskip

{\textbf{MSC2010 subsject classification}}: 91B70, 93E20, 60G40, 60G51.

\smallskip

{\textbf{JEL classification}}: C02, E22, D92, G31.

\section{Introduction}
\label{introduction}

Investment problems under uncertainty have received increasing attention in the last years in both the economic and the mathematical literature (see, for instance, \cite{DixitPindyck} for an extensive review).
Several economic papers tackle the problem of a firm maximizing
profits when the operating profit function depends on an exogenous
stochastic shock process reflecting the changes in, e.g.,
technologically feasible output, demand, and macroeconomic conditions
and so on (see, e.g., \cite{AbelEberly}, \cite{BentolilaBertola}, \cite{Bertola}, and \cite{Pindyck}), and relate irreversible investment decisions and their timing to real options (cf.\ \cite{McDonaldSiegel} and \cite{Pindyck}, among others). Usually in those models profit functions are of separable type (as Cobb-Douglas) and the economic shock process is a geometric Brownian motion.

In the mathematical-economic literature, problems of continuous-time irreversible investment under uncertainty are usually modeled as concave (or convex) stochastic control problems with monotone controls (see, e.g., \cite{Chiarolla2}, \cite{Ferrari2014}, \cite{AOksendal}, \cite{Pham} and \cite{RiedelSu}).
In fact, due to the economic constraint that does not allow
disinvestment, an irreversible investment problem under uncertainty
may be seen as a so-called `monotone follower' problem; that is, a
problem in which the investment strategies are given by nondecreasing
stochastic processes, whose associated random Borel measures on $\mathbb{R}_+$ may be singular with respect to the Lebesgue measure.
In this setting the connection between irreversible investment under uncertainty and real options found in the economic literature (cf., e.g., \cite{McDonaldSiegel} and \cite{Pindyck}) may be seen as the well known connection between concave (or convex) stochastic control problems with monotone controls and certain problems in optimal stopping. This kind of connections have been firstly rigorously shown in \cite{KaratzasElKarouiSkorohod}, \cite{Karatzas81} and \cite{KaratzasShreve84}.

When the stochastic process $X$ underlying the optimization is Markovian, e.g., a diffusion or a L\'{e}vy process, the optimal control policy usually consists in splitting the state space of the singular stochastic control problem into two regions by a curve, called the optimal investment boundary or the free-boundary. These regions are generally called `action' and `inaction' regions as it is optimal therein to turn the control on or off, respectively. Similarly in optimal stopping problems one has `continuation' and `stopping' regions where it is optimal to let the evolution of $X$ continue and cease, respectively. The main feature of concave (or convex) singular stochastic control problems is that the action region of the singular stochastic control problem coincides with the stopping region of a suitably associated optimal stopping problem and the optimal policy is to keep the controlled process just inside the inaction (continuation) region, with minimal control. 
It is then evident that a study of the optimal stopping problem
associated with a singular stochastic control problem through a
characterization of its free-boundary separating the stopping and
continuation regions leads to a complete understanding of the optimal
control, i.e., the optimal investment policy of the  firm.

In this paper we essentially consider the same problem as in
\cite{Ferrari2014} but now the economic shock is modeled by an
exponential L\'{e}vy process rather than a regular linear
diffusion. L\'{e}vy processes may exhibit heavy tails and
skewness in the probability distributions commonly found in
time-series from the market data. We solve the irreversible
investment problem by a stochastic first order conditions approach in
the spirit of \cite{BankRiedel1}, \cite{Bank}, \cite{RiedelSu}, among
others, and by relying on a suitable application of the Bank-El Karoui
representation theorem (cf.\ \cite{BankElKaroui}, Theorem 3). As in
\cite{Ferrari2014} we prove that the unique optional solution of the
Bank-El Karoui representation problem is closely linked to the
free-boundary of the one-dimensional, infinite time-horizon,
parameter-dependent optimal stopping problem naturally associated with
the original singular control problem. Such a relation and the
Wiener-Hopf factorization for L\'{e}vy processes enable us to derive
an integral equation for the free-boundary. If the underlying
  L\'{e}vy process hits any point in $\RR$ with positive probability (as
  $\alpha$-stable L\'{e}vy processes with $\alpha \in (1,2)$ or
  jump-diffusion processes which play an important role in Financial
  Economics) the free-boundary is then proved to be a unique solution
of another --  more tractable and handy -- equation. Using the equation we, moreover, prove that the free-boundary is continuous in our general L\'{e}vy process framework. To the best of our knowledge this result appears here for the first time. Finally, we find the explicit form of the optimal boundary even in the non-separable case of a CES (constant elasticity of substitution) operating profit function (see Section \ref{CESsubsection} below), thus leading to a complete characterization of the optimal investment policy of a quite intricate stochastic irreversible investment problem.

The issue of determining the optimal investment boundary of investment
problems under L\'{e}vy uncertainty has been recently tackled also
  in \cite{Boy2004} and \cite{Liangetal}. However the setting therein
is simpler than ours since in \cite{Boy2004} only separable running profits are
considered  and  in \cite{Liangetal} only a one-dimensional model is addressed. Here instead we allow any concave running profit satisfying Assumption \ref{AssProfit} below and our irreversible investment problem is set up as a two-dimensional degenerate singular stochastic control problem.
Moreover, our equation follows from the strong Markov property and the
Wiener-Hopf factorization and is not obtained by writing down any
integro-differential free-boundary problem (as is done in
\cite{Liangetal}) nor by imposing any regularity condition of the
value function 
of the associated optimal stopping problem at the
boundary itself. In this sense our approach seems to bypass the
difficulties related to the validity of the smooth-fit condition
in a
L\'{e}vy setting 
(see, e.g., \cite{KyprianouAlili}, \cite{Boy} and \cite{PeskShir2000}).

The paper is organized as follows. In Section \ref{problem} we set up
the irreversible investment problem, which is then solved in Section
\ref{optimalsolutionFB}. In Section \ref{baseandfree} we obtain a
characterization of the optimal investment boundary in terms of a base
capacity process. Section \ref{sectionintegralequation} is devoted
to the equations characterizing  the optimal investment
boundary. Finally, some examples allowing explicit calculations are
presented  in Section \ref{Examples}.


\section{The Optimal Investment Problem}

\subsection{Setting and Basic Assumptions}
\label{problem}

As in \cite{Ferrari2014} consider the optimal irreversible investment
problem of a firm producing a single good. However, to take into
account the fact that empirically the market often exhibits
significant skewness and kurtosis we model the uncertain status of the
economy (e.g., the demand of the good, or the  macroeconomic
conditions, or the price of the produced good) by the exponential
random process $e^X$, where $X=\{X_t, t \geq 0\}$ is a real valued L\'{e}vy process (other than a compound Poisson process or a subordinator)
defined on a complete filtered probability space $(\Omega,
\mathcal{F}, \{\mathcal{F}_t\}_{t\geq 0},\mathbb{P})$. For any $x\in
\mathbb{R}$, we let 
$\mathbb{P}_x(\cdot):=\mathbb{P}(\cdot | X_0=x)$ and $\mathbb{E}_x$ the corresponding expectation operator. In the following, we will simply write $\mathbb{P}_0=\mathbb{P}$ and $\mathbb{E}_0=\mathbb{E}$.

A L\'{e}vy process is a stochastic process with stationary and
independent increments, having a.s.\ c\`{a}dl\`{a}g paths
(right-continuous with left limits), and starting from zero at time
zero. Each L\'{e}vy process is fully characterized by its L\'{e}vy
triplet $(\gamma, \sigma, \Pi)$, where $\gamma, \sigma \in
\mathbb{R}$ 
and $\Pi$ is the so called L\'{e}vy measure which is concentrated on $\mathbb{R}\setminus\{0\}$ and satisfies
$$\int_{\mathbb{R}}(1 \wedge x^2)\Pi(dx) < \infty.$$

\noindent Moreover, each L\'{e}vy process $X$ can be represented as
\beq
\label{decomposition}
X_t = \gamma t + \sigma B_t + X^{(1)}_t + X^{(2)}_t,
\eeq
where $B$ is a standard Brownian motion, $X^{(1)}$ is a  zero mean pure jump martingale, $X^{(2)}$ is a compound Poisson process with jumps at least of size one, and all the components in \eqref{decomposition} are independent.
As a consequence of stationary and independent increments, it can also be shown that  
\beq
\label{Levyexponent}
\mathbb{E}[e^{i\theta X_t}] = e^{-t\Psi(\theta)},
\eeq
for all $t\geq 0$ and $\theta \in \mathbb{R}$, where 
$$\Psi(\theta):=-\log\mathbb{E}[e^{i\theta X_1}]=i \gamma \theta + \frac{1}{2}\sigma^2\theta^2 + \int_{\mathbb{R}}\Big(1 - e^{i\theta x} + i\theta x\mathds{1}_{\{|x|<1\}}\Big)\Pi(dx)$$ 
is the L\'{e}vy characteristic exponent of $X$.
Well-known L\'{e}vy processes are Brownian motion, Poisson process, jump-diffusion processes and the variance-gamma process. We refer to \cite{Bertoin} or \cite{Kyprianou} for a detailed exposition on L\'{e}vy processes.

The firm's production capacity is assumed to evolve according to
\beq
\label{productioncapacity}
C^{y,\nu}_t:= y + \nu_t, \qquad C^{y,\nu}_0:= y \geq 0,
\eeq
where $\nu$ is an (irreversible) investment plan, i.e., a nondecreasing, left-continuous, $(\mathcal{F}_t)$-adapted process such that $\nu_0=0$ $\mathbb{P}$-a.s.
 
The instantaneous profits of the firm are described by the operating
profit function $\pi: \mathbb{R}_+ \times \mathbb{R}_{+} \mapsto
\mathbb{R}_{+}$, depending on the the status of economy and on the
production capacity. The following assumption is taken to be valid
throughout the paper:
\begin{Assumptions}
\label{AssProfit}
\hspace{10cm}
\begin{enumerate}
\item The mapping $(z,c)\mapsto \pi(z,c)$ is continuous. Moreover, $c \mapsto \pi(z,c)$ is strictly increasing and strictly concave with continuous and strictly decreasing derivative $\pi_{c}(z,c):=\frac{\partial}{\partial c}\pi(z,c)$ on $\mathbb{R}_+ \times (0,\infty)$ satisfying $$\lim_{c \rightarrow 0}\pi_{c}(z,c)= \infty,\,\,\,\,\,\,\,\,\,\,\,\,\,\,\,\lim_{c \rightarrow \infty}\pi_{c}(z,c)= \kappa,$$
for some $0 \leq \kappa < \infty$.
\item The process $(\omega,t) \mapsto \pi_c\big(e^{x + X_t(\omega)}, y\big)$ is $\mathbb{P}(d\omega) \otimes e^{-rt}dt$ integrable for any $y > 0$.
\end{enumerate}
\end{Assumptions}

\noindent 
Here $r$ is a positive discount factor satisfying 
\begin{Assumptions}
\label{rbiggerthankappa}
$r > \kappa$.
\end{Assumptions}
\noindent Assumption \ref{rbiggerthankappa} will be needed in the next section to derive the optimal control policy (see Proposition \ref{existenceback} below). Moreover, we will see in Remark \ref{remnecessary} that Assumption \ref{rbiggerthankappa} is necessary to have a nonempty `no-investment region'.

\begin{remark}
Assumption \ref{AssProfit}.1 is satisfied with $\kappa = 0$ by the Cobb-Douglas and the logarithmic operating profit functions.
On the other hand, in the case of a CES (constant elasticity of substitution) profit function of the form $\pi(z,c)=(\alpha z^{\gamma} + (1-\alpha)c^{\gamma})^{\frac{1}{\gamma}}$, for some $\alpha \in (0,1)$ and $\gamma \in (0,1)$ (which reduces to the Cobb-Douglas operating profit, $\pi(z,c)=z^{\alpha}c^{1-\alpha}$, when $\gamma=0$) one has $\kappa = (1-\alpha)^{\frac{1}{\gamma}}$. It is also worth noticing that a CES operating profit with $\gamma < 0$ does not fulfill Assumption \ref{AssProfit}.1 because in this case $\lim_{\,c \rightarrow 0}\pi_{c}(z,c)= (1-\alpha)^{1/\gamma}$.
\end{remark}

For any investment plan $\nu$ the expected present value of the future overall net profits is defined as
\beq
\label{netprofit}
\mathcal{J}_{x,y}(\nu):=\mathbb{E}\bigg\{\int_0^{\infty} e^{-r t}\,\pi(e^{x + X_t}, C^{y,\nu}_t)\,dt - \int_0^{\infty} e^{- r t} d\nu_t \bigg\}.
\eeq
From now on we will call investment plans admissible if their present value is finite; i.e.\ if
\beq
\label{presentvalue}
\mathbb{E}\bigg\{\int_0^{\infty} e^{- r t} d\nu_t\bigg\} < \infty.
\eeq
We will denote by $\mathcal{S}_o$ the set of all admissible investment plans.
Due to \eqref{presentvalue} and the positivity of $\pi$ it holds that $\mathcal{J}_{x,y}(\nu)> -\infty$ for any $\nu \in \mathcal{S}_o$.
The firm's manager aims at picking an admissible $\nu^*$ 
such that
\beq
\label{valuefunction}
V(x,y):=\mathcal{J}_{x,y}(\nu^*)=\sup_{\nu \in \mathcal{S}_o}\mathcal{J}_{x,y}(\nu)<\infty, \qquad (x,y) \in \mathbb{R}\times \mathbb{R}_+.
\eeq

\noindent 
Since $\pi(z,\cdot)$ is strictly concave, $\mathcal{S}_o$ is convex
and $C^{y,\nu}$ is affine in $\nu$, we have that $\mathcal{J}_{x,y}(\cdot)$ is strictly concave on $\mathcal{S}_o$ as well. Consequently, if an optimal solution $\nu^{*}$ to (\ref{valuefunction}) does exist, it is unique. 
We provide the form of the optimal control in the next section.


\subsection{The Optimal Investment Strategy}
\label{optimalsolutionFB}

In this section we will solve the optimal investment problem \eqref{valuefunction}. A very general stochastic irreversible investment problem similar to ours \eqref{valuefunction} has been thoroughly studied in \cite{RiedelSu}, where the shock process is assumed to be a general progressively measurable process, or more recently in \cite{Ferrari2014}, in a diffusive setting.
It thus follows that some of the following results may be obtained by easily adapting arguments in \cite{Ferrari2014} or \cite{RiedelSu} (see also \cite{Bank} and \cite{Steg}). We will state them for the sake of completeness and to have a self-contained paper, but we will only sketch their proofs referring to the literature for details.


We denote by $\mathcal{T}$ the set of all $\mathcal{F}_t$-stopping
times $\tau \in [0,\infty]$ and put $e^{-r\tau(\omega)}=0$ if $\tau(\omega) = \infty$.
Following \cite{Ferrari2014}, equation (11) and Theorem 3.2, we have the following characterization of the optimal control $\nu^*$.

\begin{proposition}
\label{FOCsthm}
Under Assumption \ref{AssProfit}, a control $\nu^{*} \in \mathcal{S}_o$ such that $\mathcal{J}_{x,y}(\nu^*) < \infty$ is the unique optimal investment strategy for problem (\ref{valuefunction}) if and only if the following first order conditions for optimality
\beq
\label{FOCs}
\left\{
\begin{array}{ll}
\displaystyle \mathbb{E}\bigg\{\int_{\tau}^{\infty} e^{- r s}
\pi_c(e^{x + X_s}, C^{y,\nu^*}_s)\,ds
\Big|\mathcal{F}_{\tau}\bigg\} - e^{-r \tau} \leq 0,\quad
             {\text a.s.}\,\,\forall \tau \in \mathcal{T}, \\ \\
\displaystyle \mathbb{E}\bigg\{\int_0^{\infty}\bigg[\mathbb{E}\bigg\{\int_{t}^{\infty} e^{- r s} \pi_c(e^{x + X_s}, C^{y,\nu^*}_s)\,ds \Big|\mathcal{F}_{t}\bigg\} - e^{-r t}\bigg] d\nu^{*}_t\bigg\} = 0,
\end{array}
\right.
\eeq
hold true.
\end{proposition}
First order conditions \eqref{FOCs} may be seen as a stochastic,
infinite-dimensional generalization of the Kuhn-Tucker conditions from
the classical optimization theory. The left-hand side of the
inequality in the first condition \eqref{FOCs} is called the
supergradient process (cf.\ \cite{Ferrari2014}, equation (11) and
Remark 3.1). It is interpreted as the expected present value of the future overall net marginal profits resulting from an extra unit of investment at time $\tau \in \mathcal{T}$.
The intuition behind \eqref{FOCs} is that when the supergradient is positive at some stopping time, a small extra investment is profitable. On the other hand, investment should not occur when the supergradient is negative, since similarly reducing such an investment would be beneficial. As in \cite{Ferrari2014}, Section 3, or \cite{RiedelSu}, Theorem 3.2, the next proposition links the optimal control $\nu^*$ to the solution of a suitable Bank-El Karoui's representation problem, see \cite{BankElKaroui}, Theorem 1, Theorem 3 and Remark 2.1, related to (\ref{FOCs}).

\begin{proposition}
\label{existenceback}
Let Assumptions \ref{AssProfit} and \ref{rbiggerthankappa} hold. Then the equation
\beq
\label{backward}
\mathbb{E}\bigg\{\int_{\tau}^{\infty} e^{- r s} \pi_c\Big(e^{x + X_s}, \sup_{\tau \leq u < s}l_u\Big)\,ds \Big|\mathcal{F}_{\tau}\bigg\} = e^{-r \tau}, \qquad \tau \in \mathcal{T},
\eeq
has a unique (up to indistinguishability) strictly positive optional
solution with upper right-continuous paths\footnote{According to
  \cite{BankElKaroui}, Lemma $4.1$ (see also \cite{BankFollmer},
  Remark $1.4$-(ii)), we call a real valued process $\xi$ upper
  right-continuous on $[0,T)$ if, for each $t$, $\xi_t = \limsup_{s
      \searrow t} \xi_s$ with $\limsup_{s \searrow t} \xi_s :=
    \lim_{\epsilon \downarrow 0} \sup_{s \in [t, (t+\epsilon) \wedge
        T]} \xi_s.$}. Let $l^{*}$ denote this solution and define
\beq
\label{optimalsol}
\nu^{*}_t:=(\sup_{0 \leq s < t} l^{*}_s - y ) \vee 0,\quad t>0, \qquad \nu^{*}_0:=0.
\eeq 
If $\nu^{*}$ is admissible and such that $\mathcal{J}_{x,y}(\nu^*) < \infty$
then it is the unique optimal irreversible investment plan for problem
\eqref{valuefunction}.
\end{proposition}
\begin{proof}
We only sketch roughly the two main steps of the proof and refer to \cite{Ferrari2014} and \cite{RiedelSu} for details.\medskip

\noindent\emph{Step 1.}\,\,
Here the objective is to  prove that \eqref{backward} admits a unique (up to indistinguishability) strictly positive optional solution $l^{*}$ with upper right-continuous paths. To this end, for $\kappa$ as in Assumption \ref{AssProfit}, apply the Bank-El Karoui Representation Theorem (cf.\ \cite{BankElKaroui}, Theorem 3 and Remark 2.1) with
\begin{equation}
\label{identification}
\hat{T}=+\infty,\,\,\,\,\,\,\,\,\,\,\,\,\,\,
\mu(\omega,dt):= e^{-rt}dt
\end{equation}
and
\begin{equation}
\label{identificationf}
f(\omega,t,l):=
\left\{
\begin{array}{ll}
\pi_c\left(e^{x + X_t(\omega)}, -\frac{1}{l}\right),\,\,\,\,\,\mbox{for}\,\,l<0,\\ \\
-l + \kappa\,,\,\,\,\,\,\,\,\,\,\,\,\,\,\,\,\,\,\,\,\,\,\,\,\,\,\,\,\,\,\,\,\,\mbox{for}\,\,l\geq 0,
\end{array}
\right.
\end{equation}
to represent the deterministic process $\{e^{-rt},\, t\geq 0\}$, and then use the same arguments as in the proof of Proposition 3.4 in \cite{Ferrari2014}.\medskip

\noindent\emph{Step 2.}\,\,
Proceeding as in the proof of Theorem $3.2$ in \cite{RiedelSu}, it is easy to see that $\nu^*$ of \eqref{optimalsol} satisfies the first order conditions \eqref{FOCs}. Hence by Proposition \ref{FOCsthm} $\nu^*$ is optimal if it is admissible and such that $\mathcal{J}_{x,y}(\nu^*) < \infty$.
\end{proof}

\begin{remark}
\label{admissibility}
Notice that $\nu^*$ defined in \eqref{optimalsol} is clearly increasing and left-continuous. Moreover, since $l^*$ is optional and hence progressively measurable, then $\nu^*$ is progressively measurable by \cite{DM}, Theorem IV.33, and hence $(\mathcal{F}_t)$-adapted.
Therefore to prove that $\nu^*$ is admissible it thus remains to show
that $\nu^*$ satisfies \eqref{presentvalue}. Such a condition, as well
as the fact that $\mathcal{J}_{x,y}(\nu^*)< \infty$, is usually true
if the discount factor $r$ is big enough. In many cases this can
  be verified once the explicit form of $\nu^*$ is known (see Section \ref{Examples} for examples).
\end{remark}

\subsection{The Base Capacity and the Free-Boundary}
\label{baseandfree}

It is easy to see that our optimal policy \eqref{optimalsol} coincides
with that in Theorem 3.2 of \cite{RiedelSu} when $\delta=0$
therein. Following Definition 3.1 in \cite{RiedelSu}, the process
$l^*$ is a \textsl{base capacity} process, i.e., an index
 describing the desirable level of capacity at time $t$.
At times $t$, when the firm's production capacity is strictly above
$l^*_t$, it is optimal to wait as at those times the firm faces excess
of capacity. On the other hand, when the capacity level is below
$l^*_t$, the firm should instantaneously invest to reach the level
$l^*_t$. It therefore represents the maximal capacity level for which
it is not profitable to delay investment to any future time. Clearly, $l^*$ must be linked to the optimal boundary of an
associated optimal timing problem. Such a connection has been recently
shown in \cite{Ferrari2014}, Theorem 3.9, in a diffusive setting (see
also \cite{CF} in the context of a one-dimensional irreversible
investment problem over a finite time-horizon). In this section it is
seen that a similar connection also holds in our L\'{e}vy setting.

Similarly as in \cite{Ferrari2014}, eq.\ (25), introduce the optimal stopping problem: find a stopping time $\tau^*$ such that for all $(x,y) \in \mathbb{R} \times (0,\infty)$
\beq
\label{v}
\hspace{-0.11cm} v(x,y):=\mathbb{E}\bigg\{\int_0^{\tau^*} e^{-rs}\pi_c(e^{x + X_s},y)\,ds + e^{-r \tau^*}\bigg\}=\inf_{\tau \geq 0}\mathbb{E}\bigg\{\int_0^{\tau} e^{-rs}\pi_c(e^{x + X_s},y)\,ds + e^{-r \tau}\bigg\}.
\eeq Notice that problem \eqref{v} is the optimal timing problem associated to the irreversible investment problem \eqref{valuefunction} since it may be interpreted as the minimal cost of not investing. Mathematically, problem \eqref{v} is the one-dimensional, infinite time-horizon, parameter-dependent (as $y$ enters only as a parameter) optimal stopping problem associated to the singular control problem \eqref{valuefunction}. In fact, it can be shown (see, e.g., \cite{KaratzasBaldursson}, \cite{KaratzasElKarouiSkorohod} and \cite{KaratzasShreve84}) that under our assumptions $V_y(x,y)=v(x,y)$ and that $\tau^*:=\inf\{t \geq 0: \nu^*_t > 0\}$, with $\nu^*$ optimal for \eqref{valuefunction}, is an optimal stopping time for \eqref{v}.

Since $v(x,y) \leq 1$, for all $(x,y) \in \mathbb{R} \times (0,\infty)$, the state space splits into 
\begin{equation*}
\label{continuationstopping}
\mathcal{S}:= \{(x,y) \in \mathbb{R} \times (0,\infty): v(x,y) = 1\}, \qquad \mathcal{C}:=\{(x,y) \in \mathbb{R} \times (0,\infty): v(x,y) < 1\}.
\end{equation*}
Intuitively $\mathcal{S}$ is the region in which it is optimal to
invest immediately (the so-called `action region' or `investment
region') as therein the marginal value $v=V_y$ equals the marginal
cost of the investment. On the other hand, $\mathcal{C}$ is the region in which it is profitable to delay the investment option (the so-called `inaction region' or `no-investment region'), as the marginal value $v=V_y$ is strictly less than the marginal cost of investment therein.
Since $\pi(z,\cdot)$ is strictly concave, the mapping $y \mapsto v(x,y)$ is decreasing for any $x \in \mathbb{R}$, and therefore
\beq
\label{boundary}
b(x):= \sup\{y > 0:\,v(x,y)=1\}, \qquad x \in \mathbb{R},
\eeq
is the boundary between the stopping and continuation regions,
i.e.\ the so called free-boundary. We adopt the convention $b \equiv 0$ if $\{y > 0:\,v(x,y)=1\}=\emptyset$.

\begin{remark}
\label{remnecessary}
Notice that Assumption \ref{rbiggerthankappa} is necessary to have nonempty no-investment region $\mathcal{C}$. Indeed, if $r \leq \kappa$ then 
\begin{equation*}
v(x,y) =  1 + \inf_{\tau \geq 0}\mathbb{E}\bigg\{\int_0^{\tau} e^{-rs} \Big(\pi_c(e^{x + X_s},y) - r\Big)\,ds\bigg\} \geq 1 + \inf_{\tau \geq 0}\mathbb{E}\bigg\{\int_0^{\tau} e^{-rs} (\kappa - r)\,ds\bigg\} = 1, \nonumber 
\end{equation*}
where the inequality above is due to the fact that $\pi_c(z,c) \geq
\lim_{c \rightarrow \infty}\pi_c(z,c) = \kappa$, $(z,c) 
$.
It thus follows that if $r \leq \kappa$ then $v(x,y) = 1$ for all $(x,y) \in \mathbb{R} \times (0,\infty)$, thus implying $\mathcal{C} = \emptyset$.
\end{remark}

As in \cite{Ferrari2014}, Assumption 3.6, or \cite{RiedelSu}, Section 5, we now assume that marginal profits are positively affected by improving market conditions.
\begin{Assumptions}
\label{xpicnondecr}
The mapping $z \mapsto \pi_c(z,c)$ is nondecreasing for any $c \in (0,\infty)$.
\end{Assumptions}

\begin{remark}
\label{remsupermodular}
Notice that, if $\pi$ were twice-continuously differentiable, then
Assumption \ref{xpicnondecr} would be equivalent to requiring $\pi$ to
be supermodular (see \cite{Topkis}); that is, for any $z_1, z_2 \in
\mathbb{R}_+$ and $c \in (0,\infty)$
$$
\pi(z_1 \vee z_2, c) + \pi(z_1 \wedge z_2,c) \geq \pi(z_1,c) +
\pi(z_2,c).
$$
The Cobb-Douglas and the CES profit functions are well known examples of supermodular profit functions on $(0,\infty) \times (0, \infty)$.
\end{remark}

\begin{proposition}
\label{vcontinua}
Let Assumptions \ref{AssProfit} and \ref{xpicnondecr} hold. Then, the value function $v$ of optimal stopping problem \eqref{v} is 
\begin{itemize}
	\item continuous on $\mathbb{R} \times (0, \infty)$,
	\item such that $x \mapsto v(x,y)$, $y \in (0,\infty)$, is nondecreasing.
\end{itemize}
\end{proposition}
\begin{proof}

For the continuity, consider a sequence $\{(x_n,y_n):\, n\in\mathbb{N}\}
\subset \mathbb{R} \times (0,\infty)$ converging to $(x,y) \in \mathbb{R} \times (0,\infty)$. Take $\varepsilon>0$ and let
$\tau^\varepsilon:=\tau^\varepsilon(x,y)$ be an $\varepsilon$-optimal
stopping time for problem  \eqref{v} with initial values $x$ and $y$.
Then we have
\beq
\label{pasd}
v(x,y)-v(x_n,y_n) \geq \mathbb{E}\bigg\{\int_0^{\tau^\varepsilon} e^{-rt} \big[\pi_c(e^{x + X_t},y)-\pi_c(e^{x_n + X_t},y_n)\big] dt\bigg\} - \varepsilon.
\eeq
Without loss of generality, let $\{x_n:\, n\in\mathbb{N}\}
\subset (x-\epsilon, x+\epsilon)$ for a suitable $\epsilon >0$ be such
that for all $t\geq 0$ 
$$e^{x - \epsilon + X_t} \leq  e^{x_n + X_t} \leq
e^{x + \epsilon + X_t}.
$$  
Taking into account Assumptions \ref{AssProfit} and \ref{xpicnondecr}, we can apply the dominated convergence theorem on the right-hand side of \eqref{pasd} to get
\beq
\label{eq:cont0}
\limsup_{n\rightarrow\infty}v(x_n,y_n)\leq v(x,y)+\varepsilon.
\eeq
Similarly, taking $\varepsilon$-optimal stopping times
$\tau_{n}^\varepsilon:=\tau^\varepsilon(x_n,y_n)$  for problem \eqref{v} 
with initial values $x_n$ and $y_n$
one has
\begin{eqnarray}
\label{limit2}
v(x,y)-v(x_n,y_n) \hspace{-0.25cm} & \leq & \hspace{-0.25cm} \mathbb{E}\bigg\{\int_0^{\tau_{n}^\varepsilon} e^{-rt} \big[\pi_c(e^{x + X_t},y)-\pi_c(e^{x_n + X_t},y_n)\big] dt\bigg \} + \varepsilon \\
\hspace{-0.25cm} & \leq & \hspace{-0.25cm} \mathbb{E}\bigg\{\int_0^{\infty} e^{-rt} \big|\pi_c(e^{x + X_t},y)-\pi_c(e^{x_n + X_t},y_n)\big| dt\bigg\} + \varepsilon.\nonumber
\end{eqnarray}
Evoking again the dominated convergence theorem yields
\beq
\label{eq:cont1}
\liminf_{n\rightarrow\infty}v(x_n,y_n)\geq v(x,y)-\varepsilon,
\eeq
which together with \eqref{eq:cont0} implies the continuity of $v$.

To verify the second statement, let $\tau^*:=\tau^*(x_1,y)$ be an optimal stopping time with initial values $x_1$ and $y.$ Then for $x_2 < x_1$ we have 
\beq
\label{pasdd}
\nonumber 
v(x_1,y)-v(x_2,y) \geq \mathbb{E}\bigg\{\int_0^{\tau^*} e^{-rt}
\big[\pi_c(e^{x_1 + X_t},y)-\pi_c(e^{x_2 + X_t},y)\big]
dt\bigg\}\geq 0,
\eeq
since $\pi_c(\cdot, y)$ is supposed to be nondecreasing, see Assumption \ref{xpicnondecr}, cf.\ also the proof of Proposition 3.7 in \cite{Ferrari2014}.
\end{proof}

\begin{proposition}
\label{propb}
Under Assumptions \ref{AssProfit} and \ref{xpicnondecr} the
free-boundary $b$ defined in \eqref{boundary} is 
\begin{itemize}
	\item  nondecreasing,
	\item right-continuous with left limits.
\end{itemize}
\end{proposition}
\begin{proof}
The fact that $b$ is nondecreasing can be proved similarly as in
\cite{Ferrari2014} (see the proof of Corollary 3.8).
From the monotonicity, it clearly follows that $b$ admits right
and left limits at any point.
To show that $b$  is right-continuous, fix $x \in \mathbb{R}$ and
notice that for every $\varepsilon >0$ we have again by monotonicity
of $b$ that $b(x+\varepsilon) \geq b(x)$, which implies $b(x) \leq
\lim_{\varepsilon \downarrow 0}b(x+\varepsilon)=:b(x+)$. Consider now
the sequence $\{(x+\varepsilon, b(x+\varepsilon)):\ \varepsilon > 0\}
\subset \mathcal{S}$; one has $\{(x+\varepsilon,
b(x+\varepsilon)):\ \varepsilon > 0\} \rightarrow (x, b(x+))$ when
$\varepsilon \downarrow 0$ and $(x, b(x+)) \in \mathcal{S}$, since
$\mathcal{S}$ is closed by continuity of $v$ (cf.\ Proposition
\ref{vcontinua}). It then follows that $b(x+)\leq b(x)$ from the definition \eqref{boundary} and the proof is complete.
\end{proof}


The next theorem adapts Theorem 3.9 in \cite{Ferrari2014} to our exponential L\'{e}vy setting. It connects the base capacity process $l^*$ to the free-boundary $b$ of the optimal stopping problem \eqref{v} associated with the original control problem. 
\begin{theorem}
\label{identificocor}
Let $l^{*}$ be the unique optional solution of (\ref{backward}) and $b$ the free-boundary defined in (\ref{boundary}).
Under Assumptions \ref{AssProfit}, \ref{rbiggerthankappa} and \ref{xpicnondecr} one has
\beq
\label{identifico}
l^{*}_t = b(x + X_t), \qquad t\geq 0.
\eeq
\end{theorem}
\begin{proof}
First of all, since $b$ is Borel-measurable (being monotone) and $X$
is optional it follows that the process on the right-hand side of
\eqref{identifico} is optional. Moreover, $t\mapsto b(x + X_t)$ is
upper right-continuous since $b$ is upper-semicontinuous (being
nondecreasing and right-continuous by Proposition \ref{propb}) and $t \mapsto X_t$ is right-continuous.

To prove (\ref{identifico}), the arguments in \cite{Ferrari2014} (see
the proof of Proposition $3.4)$, are easily adapted to the present
case. Hence, we have $l^{*}_t = - \frac{1}{\xi^{*}_t}$, where the
process $\xi^{*}$ admits the following representation (cf.\ also \cite{BankElKaroui}, formula $(23)$ on page $1049$)
\begin{equation}
\label{rappresentoxistar}
\xi^{*}_t = \sup\bigg\{ l < 0 : \essinf_{\tau \geq t}\mathbb{E}\bigg\{\int_t^{\tau}e^{-r(s-t)}\pi_c\Big(e^{x + X_s}, -\frac{1}{l}\Big)\,ds + e^{-r(\tau - t)}\Big|\mathcal{F}_t\bigg\} = 1\bigg\}.
\end{equation}
To take care of the conditional expectation in \eqref{rappresentoxistar} it is convenient to proceed as in the proof of Theorem $3.9$ in \cite{Ferrari2014}, and work on the canonical probability space $(\overline{\Omega}, \overline{\mathbb{P}})$. However, to take into account our L\'{e}vy setting we let $\overline{\Omega}:=\mathcal{D}_0([0,\infty))$ be the Skorohod space of all c\`{a}dl\`{a}g functions $\overline{\omega}$ on $[0,\infty)$ such that $\omega_0=0$, endowed with Skorohod's topology and let $\mathcal{F}$ denote its Borel $\sigma$-field. Moreover, $\overline{\mathbb{P}}$ is the probability measure on $\overline{\Omega}$ under which the coordinate process $X_u(\overline{\omega})=\overline{\omega}_u$, $u\geq 0$, is a L\'{e}vy process and the shift operator $\theta_u: \overline{\Omega} \mapsto \overline{\Omega}$ is defined by $\theta_u(\overline{\omega})(s) = \overline{\omega}_{u+s}$, for $\overline{\omega} \in \overline{\Omega}$ and $u,s\geq 0$. Finally, we denote by $(\mathcal{F}_u)_{u \geq 0}$ the filtration, where $\mathcal{F}_u$ is generated by $s \mapsto \overline{\omega}_s$, $s\leq u$, and augmented by the $\overline{\mathbb{P}}$-null sets. 
By Theorem 103, p.\ 151 in \cite{DM} -- based on Galmarino's test -- any stopping time $\tau \in \mathcal{T}$, $\tau \geq t$, can be written as $\tau(\overline{\omega})= t + \tau'(\overline{\omega}, \theta_t(\overline{\omega}))$, with $\tau': \overline{\Omega} \times \overline{\Omega} \mapsto [0,\infty]$, $\mathcal{F}_t \otimes \mathcal{F}_{\infty}$-measurable and such that $\tau'(\overline{\omega}, \cdot)$ is a stopping time for each $\overline{\omega} \in \overline{\Omega}$. 
In this way, defining the $\mathcal{F}_t \otimes \mathcal{F}_{\infty}$-measurable positive random variable
\beq
\label{ZMarkov}
Z(\omega,\omega^{'}) := \int_0^{\tau'(\omega,\omega^{'})}e^{-ru}\pi_c\Big(e^{x + \omega^{'}_u}, -\frac{1}{l}\Big)\,du + e^{-r\tau'(\omega, \omega^{'})},
\eeq
and setting $Z^t(\overline{\omega}):=
Z(\overline{\omega},\theta_t(\overline{\omega}))$, after a simple
change of variable the term inside the conditional expectation in
\eqref{rappresentoxistar} equals $Z^t(\overline{\omega})$. More precisely,
\begin{eqnarray*}
&&\hskip-.5cm
\int_t^{\tau(\overline{\omega})}e^{-r(s-t)}\pi_c\Big(e^{x + X_s(\overline{\omega})}, -\frac{1}{l}\Big)\,ds + e^{-r(\tau(\overline{\omega}) - t)}
\\
&&\hskip3cm
=\int_0^{\tau'(\overline{\omega},\theta_t(\overline{\omega} ))}{\text
  e}^{-ru}\pi_c\Big(e^{x + \theta_t(\overline{\omega})(u)},
-\frac{1}{l}\Big)\,du + {\text
  e}^{-r\tau'(\overline{\omega},\theta_t(\overline{\omega}))}
\\
&&\hskip3cm
 =
Z^t(\overline{\omega}).
\end{eqnarray*}
An application of the strong Markov property (see, e.g., Exercise 3.19 at p.\ 111 of \cite{RevuzYor}) thus implies
$$\mathbb{E}\{Z^t\,|\,\mathcal{F}_t\}(\overline{\omega}) =
\mathbb{E}_{ X_t(\overline{\omega})}\{Z(\overline{\omega}, \cdot)\},$$
for all $\overline{\omega} \in \overline{\Omega}.$ Recalling the
definition of $v$ in (\ref{v}) and using \eqref{ZMarkov} it holds
  for all 
$\overline{\omega} \in \overline{\Omega}$
$$\xi^{*}_t(\overline{\omega})=\sup\{ l < 0\,:\,\,v(x + X_t(\overline{\omega}), -\frac{1}{l}) = 1\}.$$
Finally, employing arguments as those  in \cite{Ferrari2014} (see
the proof of Theorem 3.9)  we may write for each $\overline{\omega}
\in \overline{\Omega}$ and $t \geq 0$ 
$$l^{*}_t(\overline{\omega}) = -
\frac{1}{\xi^{*}_t(\overline{\omega})} = \sup\{ y > 0\,:\,\,v(x +
X_t(\overline{\omega}), y) = 1\} = b(x + X_t(\overline{\omega})),$$ 
where the last equality above follows from (\ref{boundary}). This
completes the proof.

\end{proof}

At this point it is clear that if $\nu^*$ of \eqref{optimalsol} is
admissible and such that $\mathcal{J}_{x,y}(\nu^*)<\infty$, hence
optimal, then the free-boundary $b$ of the optimal stopping problem
\eqref{v} is indeed the optimal investment boundary of problem
\eqref{valuefunction}. Moreover, then also the continuation region $\mathcal{C}$ and the stopping region $\mathcal{S}$ are the inaction and the action region, respectively.
In fact, due to Theorem \ref{identificocor} the optimal control $\nu^*$ of \eqref{optimalsol} can be expressed as
\beq
\label{optimalsol-b}
\nu^{*}_t= \sup_{0 \leq s < t}(b(x + X_s) - y) \vee 0,\quad t>0, \qquad \nu^{*}_0=0;
\eeq
i.e., it is the least effort needed at time $t$ to reflect the production capacity at the (random) time-dependent boundary $l^{*}_t=b(x + X_t)$, $t \geq 0$.

\begin{remark}
Combining Theorem \ref{identificocor} and Proposition \ref{propb} we recover \cite{RiedelSu}, Theorem 5.1, in which, via an argument different than ours, it is shown that the base capacity is monotonically increasing in the underlying shock process; namely, if $l^*$ is the base capacity associated with a L\'{e}vy process $X$ and $\tilde{l}^*$ is the base capacity associated with another L\'{e}vy process $\widetilde{X}$ such that $\widetilde{X}_t \leq X_t$ for all $t\geq 0$ a.s., then $\tilde{l}^*_t \leq l^*_t$ for all $t\geq 0$ a.s.
\end{remark}


\section{The Equation for the Optimal Investment Boundary}
\label{sectionintegralequation}

In this section we present our main result. Using Propositions \ref{existenceback} and \ref{identificocor} and the Wiener-Hopf factorization we firstly derive an integral equation for the optimal investment boundary $b$ (see Theorem \ref{integraleqthm} below). It is shown that if the L\'{e}vy process hits every point in $\mathbb{R}$ with positive probability then this equation has a unique solution.
In Theorem \ref{algebraiceqthm} another simpler equation is presented which anyway characterizes the optimal investment boundary. Using such equation we, moreover, show that the boundary is continuous. To the best of our knowledge a proof of the continuity of the free-boundary in infinite time-horizon, one-dimensional parameter-dependent optimal stopping problems of type \eqref{v} for exponential L\'{e}vy processes appears here for the first time.


To simplify exposition, in the rest of this section we will assume that $\nu^*$ of \eqref{optimalsol} is admissible and such that $\mathcal{J}_{x,y}(\nu^*)<\infty$, and hence optimal. This way the free-boundary $b$ of the optimal stopping problem \eqref{v} is indeed the optimal investment boundary of problem \eqref{valuefunction}.

\begin{theorem}
\label{integraleqthm}
Let Assumptions \ref{AssProfit}, \ref{rbiggerthankappa} and \ref{xpicnondecr} hold. Let $M_t:=\sup_{0\leq u \leq t}X_u$, $I_t:=\inf_{0\leq u \leq t}X_u$. Moreover, let $T_r$ denote an exponentially distributed random time with parameter $r$ independent of $X$. 
Then, the optimal investment boundary $b$ between the inaction and the action region is a positive nondecreasing right-continuous with left limits solution to the integral equation
\beq
\label{integraleq}
\int_0^{\infty}\mathbb{E}\Big\{\pi_c\Big(e^{y + z + I_{T_r}}, f(y + z)\Big)\Big\}\mathbb{P}(M_{T_r} \in dz) = r.
\eeq
Moreover, if the L\'{e}vy process $X$ hits every point of $\mathbb{R}$ with positive probability, then the solution of \eqref{integraleq} is unique.
\end{theorem}
\begin{proof}
Since $l^{*}$ solves (\ref{backward}) and $l^{*}_t = b(x + X_t)$ (cf.\ Theorem \ref{identificocor}), then $b$ satisfies
\begin{eqnarray}
\label{eqint1}
r &\hspace{-0.25cm} = \hspace{-0.25cm}& \mathbb{E}\bigg\{\int_{\tau}^{\infty}r e^{-r(s-\tau)}\pi_c\Big(e^{x + X_s}, \sup_{\tau \leq u < s}b(x + X_u)\Big) ds \Big| \mathcal{F}_{\tau}\bigg\} \nonumber \\
&\hspace{-0.25cm} = \hspace{-0.25cm}& \mathbb{E}\bigg\{\int_{0}^{\infty}r e^{-rt}\pi_c\Big(e^{x + X_{\tau} + (X_{t + \tau} - X_{\tau})}, b\big(\sup_{0 \leq u < t}(x + X_{\tau} + X_{u + \tau} - X_{\tau})\big)\Big) dt \Big| \mathcal{F}_{\tau}\bigg\},
\end{eqnarray}
for any $\tau \in \mathcal{T}$, where in the second equality we have
used the fact that $b$ is nondecreasing. 
Using the independence of increments and the strong Markov property
 of $X$ it is seen that  (\ref{eqint1}) is equivalent with
\begin{equation*}
\label{eqint2}
\mathbb{E}_y\bigg\{\int_{0}^{\infty}r e^{-rt}\pi_c\Big(e^{X_t}, b(M_t)\Big) dt \bigg\} = r, \qquad \forall y \in \mathbb{R},
\end{equation*}
and, furthermore, with
\begin{equation}
\label{eqint3}
\mathbb{E}_y\Big\{\pi_c\Big(e^{X_{T_r}}, b(M_{T_r})\Big)\Big\} = r, \qquad \forall y \in \mathbb{R}.
\end{equation}
But now, by the Wiener-Hopf factorization (cf.\ \cite{Kyprianou},
Chapter 6) we know that $X_{T_r} - M_{T_r}$ is independent of
$M_{T_r}$ and $X_{T_r} - M_{T_r}$ has the same law as  $\hat{I}_{T_r}$ with $\hat{I}$ an independent copy of $I$, and then we can write from \eqref{eqint3}
\begin{eqnarray*}
\label{eqint4}
r \hspace{-0.25cm}& = &\hspace{-0.25cm} \mathbb{E}_y\Big\{\pi_c\Big(e^{X_{T_r}}, b(M_{T_r})\Big)\Big\} = \mathbb{E}\Big\{\mathbb{E}\Big\{\pi_c\Big(e^{y + X_{T_r}}, b(y + M_{T_r})\Big)\Big| M_{T_r}\Big\}\Big\} \nonumber \\
& \hspace{-0.25cm} = \hspace{-0.25cm}& \mathbb{E}\Big\{\mathbb{E}\Big\{\pi_c\Big(e^{y + M_{T_r} + (X_{T_r} - M_{T_r})}, b(y + M_{T_r})\Big)\Big| M_{T_r}\Big\}\Big\}  \\
& \hspace{-0.25cm} = \hspace{-0.25cm}&\int_0^{\infty}\mathbb{E}\Big\{\pi_c\Big(e^{y + z + I_{T_r}}, b(y + z)\Big)\Big\}\mathbb{P}(M_{T_r} \in dz), \nonumber
\end{eqnarray*}
where spatial homogeneity of $X$ has been used for the second equality above.

Finally, if  $\tau^{\{x_o\}}:=\inf\{t \geq 0: X_t=x_o\}<\infty$ with
positive probability for all $x_o \in \mathbb{R}$, then the uniqueness of
a positive, nondecreasing right-continuous with left limits $b$ satisfying
(\ref{integraleq}) can be proved arguing by contradiction as in
\cite{Ferrari2014} (see the proof of Theorem 3.11). 

\end{proof}
\begin{remark}
\label{levyregularity}
Sufficient conditions ensuring that the L\'{e}vy process $X$ hits
every point of $\mathbb{R}$ with positive probability can be found, e.g., in Theorem 7.12 of \cite{Kyprianou}. 
Examples of L\'{e}vy processes with such a property are any L\'{e}vy process with Gaussian component (including the case of jump-diffusion processes that play an important role in Financial Economics) or symmetric $\alpha$-stable L\'{e}vy processes with $\alpha \in (1,2)$ (see Section 7.5 and Exercise 7.6 in \cite{Kyprianou}). 
Concerning spectrally one-sided L\'{e}vy processes one has that the
set $C:=\{x \in \mathbb{R}: \mathbb{P}(\tau^{\{x\}} < +\infty)>0\}$
as defined in 
Theorem 7.12 of \cite{Kyprianou} is not empty (since spectrally one-sided L\'{e}vy processes creep) and therefore condition (7.21) therein is satisfied and (i)-(iii) apply accordingly.
\end{remark}

The Wiener-Hopf factorization in the proof of Theorem
\ref{integraleqthm} replaces in our L\'{e}vy setting the use of the
joint law of the position of a regular, one-dimensional diffusion and
its running supremum evaluated at an independent exponential time
(cf.\ \cite{CsakiFoldesSalminen} p.\ 185 and \cite{BorodinSalminen} p.\ 26) exploited in \cite{Ferrari2014}, Theorem 3.11, in the diffusive setting. It is also worth noticing that the Wiener-Hopf factorization has been recognized as a useful tool for solving one-dimensional, infinite time-horizon, optimal stopping problems for L\'{e}vy processes as shown in \cite{Boy}, \cite{ChrSalBao}, \cite{Deligiannidisetal}, \cite{Mordeki}, \cite{MordekiSalminen}, \cite{Salminen}, among others.

The next theorem represents our main result. It shows that in order to have a solution to \eqref{integraleq} it suffices to find a solution of a simpler equation.

\begin{theorem}
\label{algebraiceqthm}
Under Assumptions \ref{AssProfit}, \ref{rbiggerthankappa} and
\ref{xpicnondecr}, there exists a unique positive function
$\hat{b}$ satisfying the equation 
\beq
\label{algebraiceq}
\mathbb{E}\Big\{\pi_c\Big(e^{u + I_{T_r}}, f(u)\Big)\Big\} = r, \qquad u \in \mathbb{R}.
\eeq
The function $\hat b$ is nondecreasing and continuous.
Moreover, if the L\'{e}vy process $X$ hits every point of $\mathbb{R}$ with positive probability, then $\hat{b}$ is the optimal investment boundary between the inaction and the action region.
\end{theorem}

\begin{proof}

We start with showing that \eqref{algebraiceq} admits at most one positive solution $\hat{b}$ such that it is nondecreasing and continuous.
Define the function
\beq
\label{Phi}
\Phi(u,y):= \mathbb{E}\Big\{\pi_c\Big(e^{u + I_{T_r}}, y\Big)\Big\} - r, \quad (u,y) \in \mathbb{R} \times (0,\infty).
\eeq
It is not hard to see that $\Phi(u,\cdot)$ is (strictly) decreasing for any $u \in \mathbb{R}$ due to Assumption \ref{AssProfit}. Moreover, thanks again to Assumption \ref{AssProfit}, we can apply the monotone convergence theorem to show that $\Phi(u,\cdot)$ is continuous on $(0,\infty)$, $\lim_{y\downarrow 0} \Phi(u,y) = \infty$ and $\lim_{y\uparrow \infty} \Phi(u,y) = \kappa-r < 0$ for any $u \in \mathbb{R}$, where the last limit is strictly negative by Assumption \ref{rbiggerthankappa}.
Hence there exists a unique positive $\hat{b}(u)$, $u \in \mathbb{R}$, solving \eqref{algebraiceq}.

To prove that $\hat{b}$ is nondecreasing, fix $\varepsilon > 0$ and notice that \eqref{Phi} and the fact that $\hat{b}$ solves \eqref{algebraiceq} imply
\begin{eqnarray*}
0 \hspace{-0.25cm} & = & \hspace{-0.25cm} \Phi(u + \varepsilon, \hat{b}(u + \varepsilon)) - \Phi(u, \hat{b}(u)) = \Phi(u + \varepsilon, \hat{b}(u + \varepsilon)) - \Phi(u, \hat{b}(u + \varepsilon)) \nonumber \\
&& \hspace{0.1cm} +\, \Phi(u, \hat{b}(u + \varepsilon)) - \Phi(u, \hat{b}(u)) \geq \Phi(u, \hat{b}(u + \varepsilon)) - \Phi(u, \hat{b}(u)), \nonumber
\end{eqnarray*} 
where the inequality above follows since, by Assumption \ref{xpicnondecr}, $\Phi(\cdot,y)$ is nondecreasing for any $y \in (0,\infty)$.
Hence, one has $\Phi(u, \hat{b}(u + \varepsilon)) - \Phi(u,
\hat{b}(u)) \leq 0$ and therefore $\hat{b}(u+\varepsilon) \geq
\hat{b}(u)$, $u \in \mathbb{R}$, because $ y\mapsto \Phi(u,y)$ is nonincreasing.

We next show the continuity of $\hat{b}$. We consider first its left-continuity.
Fix $u \in \mathbb{R}$, take a sequence $\{u_n:\, n \in \mathbb{N}\}
\subset \mathbb{R}$ such that $u_n \uparrow u$ as $n \uparrow \infty$,
and define $\hat{b}(u-):=\lim_{n \uparrow \infty}\hat{b}(u_n)$. Without loss of generality, we may think that
$u \geq u_n \geq u - \epsilon$, for a suitable $\epsilon > 0$, so that
$\hat{b}(u_n) \geq \hat{b}(u-\epsilon)$ by the monotonicity of
$\hat{b}$ and the fact that $e^{u_n + I_{T_r}} \leq e^u$ since $I_{T_r} \leq 0$. It thus follows from the concavity of $c \mapsto \pi(z,c)$
and the monotonicity of $z \mapsto \pi_c(z,c)$ (cf.\ Assumptions
\ref{AssProfit} and \ref{xpicnondecr}, 
respectively) that $\pi_c(e^{u_n +
  I_{T_r}}, \hat{b}(u_n)) \leq \pi_c(e^{u},
\hat{b}(u-\epsilon)).$ Note that for $\{u_n:\, n \in \mathbb{N}\}$ as above we have 
$$r = \mathbb{E}\Big\{\pi_c\Big(e^{u_n + I_{T_r}}, \hat{b}(u_n)\Big)\Big\}, \qquad n \in \mathbb{N}.$$
Letting $n\uparrow \infty$ yields
\beq
\label{left2}
r = \lim_{n \uparrow \infty}\mathbb{E}\Big\{\pi_c\Big(e^{u_n + I_{T_r}}, \hat{b}(u_n)\Big)\Big\} = \mathbb{E}\Big\{\pi_c\Big(e^{u + I_{T_r}}, \hat{b}(u-)\Big)\Big\},
\eeq
where the dominated convergence theorem and joint continuity of $\pi$ (cf.\ Assumption \ref{AssProfit}) are used.
We then conclude that $\hat{b}(u-)=\hat{b}(u)$ by the uniqueness of the solution of \eqref{algebraiceq}.
On the other hand, taking a sequence $\{u_n:\, n \in \mathbb{N}\}
\subset \mathbb{R}$ such that $u_n \downarrow u$ as $n \uparrow
\infty$ and following similar arguments as above, one can prove also
right-continuity of $\hat{b}$. Consequently, $\hat b$ is continuous.

Clearly, the positive, nondecreasing continuous $\hat{b}$ solving
\eqref{algebraiceq} also solves \eqref{integraleq}, and then the
positive, upper right-continuous optional process
$\hat{l}_t:=\hat{b}(x + X_t)$ solves \eqref{backward}. Hence, by
Proposition \ref{existenceback} and Theorem \ref{identificocor} we
have $\hat{l}_t:=\hat{b}(x + X_t) = b(x + X_t)= l^*_t$, up to the indistinguishability. The proof is completed by arguing similarly as in the proof of Theorem \ref{integraleqthm} that $\hat{b}=b$.

\end{proof}

As a remarkable by-product of Theorem \ref{algebraiceqthm} we have the following result.

\begin{corollary}
\label{cor-bcont}
If the L\'{e}vy process $X$ hits every point of $\mathbb{R}$ with positive probability, then the optimal investment boundary $b$ between the inaction and the action region is continuous.
\end{corollary}


\begin{remark}
It is worth noting that in the separable case $\pi(z,c) = z G(c)$, with $G$ continuously differentiable, increasing, strictly concave and satisfying Inada conditions, equation \eqref{algebraiceq} easily translates into equation (15) of \cite{Boy2004} (see also \cite{RiedelSu}, Example 3.3), where $H$ therein is the (generalized) inverse of $b$. However, our result is much more general than that of \cite{Boy2004}. Differently to \cite{Boy2004} we have indeed not assumed that investment strategies stay bounded above and, moreover, our equation \eqref{algebraiceq} holds for every operating profit satisfying Assumption \ref{AssProfit}, hence not necessarily separable. According to the discussion in \cite{Boy2004}, Section III, equation \eqref{algebraiceq} may be seen as a correction to the Net Present Value rule, or Marshallian law, taking into account the irreversibility of investment strategies.
\end{remark}

\begin{remark}
Combining Theorem \ref{identificocor} and Corollary \ref{cor-bcont}, we find that if $X$ hits every point of $\mathbb{R}$ with positive probability, then the base capacity process $l^*$ of problem \eqref{valuefunction}, which in general is only known to be upper right-continuous by \cite{BankElKaroui}, has indeed the same path regularity as $X$, namely it has at least c\`{a}dl\`{a}g paths.
\end{remark}

To some extent, our  equation (\ref{algebraiceq}) may be interpreted as a substitute to the free-boundary value problem which one usually writes down to
characterize the solution to an optimal stopping problem (see \cite{PeskShir} for a review).
In the case of optimal stopping problems with L\'{e}vy uncertainty it
is still possible to derive a free-boundary problem (see, e.g., \cite{Boy}), even if one has to pay attention to the sense in which
the associated integro-differential operator is understood, and to which are the suitable regularity properties of the value function to be imposed at the boundary.
It has been in fact noticed (see, e.g., \cite{KyprianouAlili}, \cite{Boy} and  \cite{PeskShir2000}) that the smooth-fit property of the value function of an optimal stopping problem (i.e.\ its $C^1$-property at the optimal boundary) may fail in a L\'{e}vy setting.
Our  equation (\ref{algebraiceq}), instead, is not derived from any
free-boundary problem but it follows immediately from
\eqref{integraleq}, thanks to the backward equation (\ref{backward})
for $l^{*}=b(x + X)$, the Wiener-Hopf factorization and the strong Markov property of $X$. It thus represents a very useful tool to determine the optimal investment boundary for the whole class of irreversible investment problems of type (\ref{valuefunction}), under the assumption that the L\'{e}vy process hits every point of $\mathbb{R}$ with positive probability. In the next section we will show how to analytically solve equation (\ref{algebraiceq}) even in the non trivial case with a non-separable profit function.


\section{Explicit Results}
\label{Examples}

In this section we derive the explicit form of the optimal investment
boundary of the irreversible investment problem \eqref{valuefunction}
for the Cobb-Douglas and the CES (constant elasticity of substitution) operating profit functions, that is, for $\pi(z,c) = z^{\alpha}c^{\beta}$ with $\alpha, \beta \in (0,1)$, and $\pi(z,c)= (\alpha z^{\gamma} + (1-\alpha)c^{\gamma})^{\frac{1}{\gamma}}$, with $\alpha, \gamma \in (0,1)$, respectively.
Moreover, we will assume throughout this section that the L\'{e}vy
process $X$ hits any point of $\mathbb{R}$ with positive probability,
so to have the optimal investment boundary as the unique solution of equation \eqref{algebraiceq} (cf.\ Theorem \ref{algebraiceqthm}).

Recall that $T_r$ is an exponentially distributed random time with
parameter $r$ independent of $X$, $M_t:=\sup_{0\leq u \leq t}X_u$ and $I_t:=\inf_{0\leq u \leq t}X_u$.
The notation $\widehat\Psi$ is used for the logarithm of the Laplace transform of $X_1$ (when well defined), i.e.,
$$
\widehat\Psi(\lambda):=\log \EE\big\{e^{\lambda X_1}\big\}.
$$

\subsection{Cobb-Douglas Operating Profit}

Assume that the operating profit function is of the Cobb-Douglas type; that is, $\pi(z,c) = z^{\alpha}c^{\beta}$ for $\alpha, \beta \in (0,1)$.

\begin{proposition}
\label{propboundaryCD}
Assume that $\widehat\Psi(\frac{\alpha}{1-\beta}) \vee \widehat\Psi(\alpha + \beta)$ is well defined and that (cf.\ also Assumption \ref{rbiggerthankappa})
\beq
\label{lapcon_1}
r> 0 \vee \widehat\Psi(\frac{\alpha}{1-\beta}) \vee \widehat\Psi(\alpha + \beta).
\eeq
Then for a Cobb-Douglas operating profit the optimal investment boundary is 
\beq
\label{boundaryCD}
b(x)=(\vartheta e^x)^{\frac{\alpha}{1-\beta}}, \qquad x \in \mathbb{R},
\eeq
with $\vartheta$ given by
\beq
\label{kappaCD}
\vartheta:=\left(\frac{\beta \mathbb{E}\{e^{\alpha I_{T_r}}\}}{r}\right)^{\frac{1}{\alpha}}.
\eeq
\end{proposition}
\begin{proof}
In this case equation \eqref{algebraiceq} has the form
\beq
\label{integraleqCD}
r=\beta e^{\alpha x}\mathbb{E}\{e^{\alpha I_{T_r}}\} b^{\beta-1}(x).
\eeq
Taking $b(x)=(\vartheta e^x)^{\frac{\alpha}{1-\beta}}$ it is easy to see that \eqref{integraleqCD} above is solved for $\vartheta$ as in \eqref{kappaCD}.

The arguments employed in  \cite{RiedelSu} (see  the proof of Theorem $7.2$) are easily adapted to our case with $\alpha \neq 1-\beta$ to show that if (\ref{lapcon_1}) holds then $\nu^*_t := \sup_{0 \leq s < t} (b(x + X_s) - y) \vee 0$, $t>0$, and  $\nu^*_0:=0$, (cf.\ \eqref{optimalsol-b}) is admissible and $\mathcal{J}_{x,y}(\nu^*) < \infty$. Therefore, $\nu^*$ is optimal for problem \eqref{valuefunction} and, hence, $b$ as given in \eqref{boundaryCD} is the optimal investment boundary.
\end{proof}

\begin{remark}
The result of Proposition \ref{propboundaryCD} is in line with the findings of Proposition 7.1 and Theorem 7.2 in \cite{RiedelSu}, in which the base capacity process $l^*$ has been explicitly determined in the case of L\'{e}vy processes and Cobb-Douglas profits.
\end{remark}

\subsection{CES Operating Profit}
\label{CESsubsection}

We turn now to the case with a non-separable operating profit of the CES (constant elasticity of substitution) type, that is, $\pi(z,c)= (\alpha z^{\gamma} + (1-\alpha)c^{\gamma})^{\frac{1}{\gamma}}$ for some $\alpha \in (0,1)$. Moreover, to meet Assumption 2.1.1 let $\gamma \in (0,1)$
to have $\lim_{c \rightarrow 0}\pi_c(z,c) = 0$ and $\kappa:= \lim_{c \rightarrow \infty}\pi_c(z,c)= (1-\alpha)^{\frac{1}{\gamma}}$.
To the best of our knowledge, the explicit form of the optimal investment boundary of problem \eqref{valuefunction} for a non-separable profit of CES type and exponential L\'{e}vy processes appears here for the first time.

\begin{proposition}
\label{propboundaryCES}
Assume that $\widehat\Psi(1)$ is well defined and that (cf.\ also Assumption \ref{rbiggerthankappa})
\beq
\label{lapcon_2}
r>(1-\alpha)^{\frac{1}{\gamma}} \vee \widehat \Psi(1).
\eeq
Then for a CES operating profit the optimal investment boundary is given by
\beq
\label{boundaryCES}
b(x)= K e^x, \qquad x \in \mathbb{R},
\eeq
where the constant $K$ (depending on $\gamma$, $\alpha$ and $r$) is the unique positive solution to
\beq
\label{betaCES}
\mathbb{E}\Big\{\Big(1 + \Big(\frac{\alpha}{1-\alpha}\Big)e^{\gamma I_{T_r}}K^{-\gamma}\Big)^{\frac{1-\gamma}{\gamma}}\Big\} = \frac{r}{(1-\alpha)^{\frac{1}{\gamma}}}.
\eeq
\end{proposition}
\begin{proof}
In this case equation \eqref{algebraiceq} becomes
\beq
\label{integraleqCES}
\frac{r}{(1-\alpha)^{\frac{1}{\gamma}}} = \mathbb{E}\Big\{\Big(1 + \Big(\frac{\alpha}{1-\alpha}\Big)e^{\gamma (x+I_{T_r})}b^{-\gamma}(x)\Big)^{\frac{1-\gamma}{\gamma}}\Big\}.
\eeq
Comparing  \eqref{integraleqCES} and  \eqref{betaCES} it is seen that $b(x)=K e^x$ is a natural candidate for the optimal
boundary. To validate our candidate we firstly have to show that \eqref{betaCES} actually admits at most one positive solution.
Define the function $F:(0,\infty) \mapsto \mathbb{R}$ as
$$F(u):=\mathbb{E}\Big\{\Big(1 + \Big(\frac{\alpha}{1-\alpha}\Big)e^{\gamma I_{T_r}}u^{-\gamma}\Big)^{\frac{1-\gamma}{\gamma}}\Big\} - \frac{r}{(1-\alpha)^{\frac{1}{\gamma}}}.$$
It is clear that $u\mapsto F(u)$ is strictly decreasing, continuous and because $0<\gamma < 1$ it holds 
$$\lim_{u \downarrow 0}F(u) \geq \lim_{u \downarrow 0} u^{\gamma - 1}\mathbb{E}\Big\{\Big(\Big(\frac{\alpha}{1-\alpha}\Big)e^{\gamma I_{T_r}}\Big)^{\frac{1-\gamma}{\gamma}}\Big\} - \frac{r}{(1-\alpha)^{\frac{1}{\gamma}}} = \infty.$$
Moreover, since $0 < \gamma < 1$ one also has $0 \leq e^{\gamma I_{T_r}}\leq 1$ and then
$$\lim_{u \uparrow \infty}F(u) \leq \lim_{u \uparrow \infty}\Big(1 + \Big(\frac{\alpha}{1-\alpha}\Big)u^{-\gamma}\Big)^{\frac{1-\gamma}{\gamma}} - \frac{r}{(1-\alpha)^{\frac{1}{\gamma}}} = 1-\frac{r}{(1-\alpha)^{\frac{1}{\gamma}}} < 0,$$
where the last inequality is due to the fact that $r > (1-\alpha)^{\frac{1}{\gamma}}$ by the assumption.
It thus follows that $F(u)=0$ admits at most one positive solution.

To complete the proof, we have to show that $\nu^*_t := \sup_{0 \leq s
  < t} (b(x + X_s) - y) \vee 0$, $t>0$, $\nu^*_0:=0$,
(cf.\ \eqref{optimalsol-b}) is admissible and such that
$\mathcal{J}_{x,y}(\nu^*) < \infty$; hence, optimal for problem
\eqref{valuefunction}. Clearly, $\nu^*$ is $(\mathcal{F}_t)$-adapted,
left-continuous and nondecreasing. By the Wiener-Hopf factorization 
\beq
\label{WHfact}
\mathbb{E}\big\{e^{ M_{T_r}}\big\}\mathbb{E}\big\{e^{I_{T_r}}\big\} = \frac{r}{r - \widehat{\Psi}(1)},
\eeq
Then recalling (\ref{lapcon_2}) and using \eqref{WHfact} we obtain
\begin{eqnarray*}
\mathbb{E}\bigg\{\int_0^{\infty}re^{-rt}\nu^*_t dt\bigg\}
\hspace{-0.25cm} & \leq & \hspace{-0.25cm} K\mathbb{E}\bigg\{\int_0^{\infty}re^{-rt} \sup_{0 \leq
    s < t} e^{x+X_s}\, dt\bigg\} = K\mathbb{E}\bigg\{\int_0^{\infty}re^{-rt} e^{x+M_{t}}\, dt\bigg\} \nonumber \\
\hspace{-0.25cm} & = & \hspace{-0.25cm}Ke^x\mathbb{E}\big\{e^{M_{T_r}}\big\} < \infty. \nonumber 
\end{eqnarray*}
Consequently, integrating by parts yields
\beq
\label{admissibilitynustarCES}
\mathbb{E}\bigg\{\int_0^{\infty}e^{-rt}d\nu^*_t\bigg\} < \infty,
\eeq
i.e., $\nu^*$ is admissible. Next consider
\begin{eqnarray}
\label{Vfinite}
&& \mathbb{E}\bigg\{\int_{0}^{\infty}e^{-rt}\pi(e^{x +
  X_t}, y + \nu^{*}_t) dt\bigg\} \nonumber\\
&&\hskip2cm  \leq \frac{1}{r}\mathbb{E}\bigg\{\int_{0}^{\infty}re^{-rt}\Big(e^{\gamma(x + X_t)} + \big(y + \sup_{0 \leq s < t}Ke^{x + X_s}\big)^{\gamma}\Big)^{\frac{1}{\gamma}} dt\bigg\} \nonumber \\
&&\hskip2cm  \leq \frac{2^{\frac{1-\gamma}{\gamma}}}{r}\mathbb{E}\bigg\{\int_{0}^{\infty}re^{-rt}\Big(e^{x + X_t} + y + \sup_{0 \leq s < t}Ke^{x + X_s}\Big) dt\bigg\} \nonumber \\
&&\hskip2cm = \frac{2^{\frac{1-\gamma}{\gamma}}}{r}\Big( y + e^x \mathbb{E}\big\{e^{X_{T_r}}\big\} + K e^x\mathbb{E}\big\{e^{M_{T_r}}\big\} \Big) < \infty,
\end{eqnarray}
where we have used again (\ref{lapcon_2}) and \eqref{WHfact}. Combining (\ref{admissibilitynustarCES}) and
(\ref{Vfinite}) shows that $\mathcal{J}_{x,y}(\nu^*) < \infty$ (cf.\ (\ref{netprofit})).

It thus follows that $\nu^*_t$ is optimal for problem \eqref{valuefunction} and $b$ in \eqref{boundaryCES} is the optimal investment boundary.
\end{proof}
\begin{remark}
Clearly, the case $\gamma=\frac{1}{n}$, $n \geq 2$, discussed in
\cite{Ferrari2014} in a diffusive setting, is a particular case of the
one studied in Proposition \ref{propboundaryCES} and it follows that
in such a case $b(x)=K e^x$ for some positive constant
$K:=K(n,\alpha,r)$ solving equation \eqref{betaCES}. An application of
the binomial expansion (see also \cite{Ferrari2014}, Section 4.2)
reduces equation \eqref{betaCES} for the constant $K$ to the following
polynomial equation of order $n-1$ 
$$\sum_{j=1}^{n-1}
\begin{pmatrix}
n-1 \\
j
\end{pmatrix} A_{j,n} \Big(\frac{\alpha}{1-\alpha}\Big)^{j}K^{-\frac{j}{n}} - \left[\frac{r}{(1-\alpha)^{\frac{1}{\gamma}}} - 1\right] = 0$$
with $A_{j,n}:=\mathbb{E}\{e^{\frac{j}{n}I_{T_r}}\}$.
Such a polynomial equation admits a unique positive solution thanks to Descartes' rule of signs since $r > (1-\alpha)^{\frac{1}{\gamma}}$.
\end{remark}


\textbf{Acknowledgments.} Giorgio Ferrari thanks Frank Riedel for useful discussions. Paavo Salminen thanks the Center for Mathematical Economics, Bielefeld University for the hospitality and the support during the stay in Bielefeld and Andreas Kyprianou for a stimulating exchange via e-mail.


\end{document}